\documentclass[A4paper]{article} 
\usepackage{amsmath}
\usepackage{amscd}
\usepackage{amsthm}
\usepackage[mathscr]{eucal}
\usepackage{amssymb} 
\usepackage{latexsym}
\usepackage{amsthm}
\theoremstyle{definition}
\newtheorem{theorem}{Theorem}[section]
\newtheorem{main}{Main theorem}
\newtheorem{lemma}[theorem]{Lemma}
\newtheorem{proposition}[theorem]{Proposition}
\newtheorem{corollary}[theorem]{Corollary}
\newtheorem{remark}[theorem]{Remark}
\newtheorem{definition}[theorem]{Definition}

\title{Remarks on low-energy approximations 
for Feynman path integration on the sphere}
\author{Yoshihisa Miyanishi}
\date{\today}
\begin{document}
\maketitle
\begin{abstract} An alternative method of the Feynman path integral 
for $S^2$ is proposed. Our method employs action integrals 
$S(t,x,y)$ along shortest paths. 
We shall define the operator 
$$U_{\chi}(t)f(x)=(2\pi i)^{-1}\int_{S^2}\chi(x, y) a(x, y, t) 
\exp{\{i S(t,x,y)+\frac{i R t}{12}\}} f(y)\;dy,$$ where  $a(x, y, t)$ 
is the square root of Van Vleck determinant, 
$\frac{R}{12}$ is Dewitt curvature
and $\chi(x, y)$ is the bump function with a small support. 
By using spectral projectors $\rho(E)$ of $-\triangle$, 
we have the uniform and strong convergences:   
\begin{align*} 
&\hspace{4mm} \lim\limits_{N \rightarrow \infty}\Vert [\{U_{\chi}(t/N)\}^N-e^{it\triangle/2}] 
\rho({N^{1/3-\varepsilon}}) \Vert_{L^2}=0 
\hspace{14.5mm} \mbox{(convergence of operator norm)}, \\ 
&{s-}\lim\limits_{N \rightarrow \infty}\{U_{\chi}(t/N)\}^N \rho(N) f(x)
=e^{it\triangle/2} f(x) \quad \mbox{in}\ L^2(S^2) 
\hspace{7mm} \mbox{(strong convergence)}. 
\end{align*}
This is a rigorous construction of Schr\H{o}dinger evolution operators.  
\end{abstract}
\renewcommand{\thefootnote}{}
\section{Introduction and results.}
\par\hspace{5mm} 
Consider the standard sphere $(S^2,\ g_{st})$. So the geodesic distance 
$d(x,y)=\arccos (x \cdot y)$ 
for $x, y\in S^2\subset \mathbf{R}^3$, where $x \cdot y$ denotes an 
inner product. 
The scalar curvature $R$ is $2$ and the Gaussian curvature  is $1$. 
The action $S(t,x,y)$ is represented as an integral over time, taken along the 
geodesic path on the sphere  
between the initial time and the final time of the development of the system: 
$$ S(t,\ x,\ y)=\int_0^t \frac{1}{2}g_{x(t)}(\dot{x}(t), \dot{x}(t))\; dt 
=\frac{|d(x, y)|^2}{2t}, $$ 
and Van Vleck determinant $V(t,x,y)$ is defined by 
$$
V(t,x,y)=g^{-1/2}(x)g^{-1/2}(y) \det _{ij}
\left(\frac{\partial^2 S(t,\ x,\ y)}
{{\partial {x_i}}{\partial {y_j}}}\right).
$$
\par A reasonable candidate for approximating operators 
of $e^{\frac{it\triangle}{2}}$ associated to $S$ and $V$ 
is represented as the oscillatory integral operators. 
\par
\begin{definition}[Shortest path approximations on $S^2$]
%Let $L(x,\dot{x})=\frac{|\dot{x}|^2}{2}
%=\frac{1}{2}g_{x(t)}(\dot{x}(t), \dot{x}(t))$ 
%be the Lagrangian function. 
The shortest path approximation 
$U_{\chi}(t)$ on $S^2$ is defined by 
\begin{align*}
U_{\chi} (t)f(x)\equiv \frac{1}{2\pi i}\int_{S^2} \chi(d(x,y)) 
\sqrt{V(t, x, y)} e^{iS(t,\ x,\ y)+\frac{iRt}{12}}f(y)\; dy,  
\end{align*}
where $\chi(d(x,y))$ is the bump function with compact support contained 
in $d(x, y)<\pi$ and     
 $\frac{R}{12}$ is Dewitt curvature constant.
\end{definition}
Under these conditions, the main result of this paper (\S 4) is the following: 
%%%%%%%%%%%%%%%%%%%%%%%%%%%%%%%%%%%%%%%%%%%%%
\begin{main}[Time slicing products and the strong limits]
For $f(x) \in  L^2(S^2)$, we have 
$$
s-\lim_{N\rightarrow \infty} \{U_{\chi}(t/N)\}^N \rho(N) f(x) 
=e^{\frac{it\triangle}{2}} f(x) \quad\mbox{in}\ L^2. 
$$ 
where $\rho(N)$ are spectral projectors defined by the spectral theorem : 
$-\triangle=\int_{0}^{\infty}E\ d\rho(E)$. 
\end{main} 
%%%%%%%%%%%%%%%%%%%%%%%%%%%%%%%%%%%%%%%%%%%%%%%%%%
If $f(x)$ is a low energy function (i.e. a finite sum of eigenfunctions of $-\triangle$), 
the covergence of time slicing products is given without spectral projectors : 
%%%%%%%%%%%%%%%%%%%%%%%%%%%%%%%%%%%%%%%%%%%%%%%%
\begin{corollary}
Let $Y_{l, m}$ be spherical harmonics and $f(x)=\sum\limits_{{\rm finite}} a_{l, m} Y_{l, m}$ be a 
finite sum of spherical harmonics. Then 
$$
s-\lim_{N\rightarrow \infty} \{U_{\chi}(t/N)\}^N f(x) 
=e^{\frac{it\triangle}{2}} f(x) \quad\mbox{in}\ L^2. 
$$ 
\end{corollary}
%%%%%%%%%%%%%%%%%%%%%%%%%%%%%%%%%%%%%%%%%%%%%%%%%%%%%%%%%%%%
This is an analogous result to the Feynman integral proposed by means 
of finite dimensional approximations and Trotter type time 
slicing products (See e.g. \cite{Fu 1},\cite{Fu 2}, \cite{Fu 3}, 
\cite{Fu-Tu}, \cite{I-W}, \cite{Ino}, \cite{Int}, \cite{Ki-KH}, 
\cite{KN}, \cite{Ya 1}). In these papers, the stationary action 
trajectories are finite for fixed time $t>0$ and $x, y\in \mathbf{R}^n$,  
and the kernel $E(t,x,y)$ of $e^{it( \frac{-\triangle}{2}+V(x))}$ 
are bounded smooth for small $t\not=0$. Thus time slicing products 
converge without spectral projectors.  
\par
On compact manifolds, however, infinite many 
action paths exist, even if time $t>0$ is fixed. 
To clarify the meaning of spectral projectors,
we consider the quantum evolution on $S^2$. 
By Stone's theorem {(See e.g. \cite{RS})}, $e^{\frac{it\triangle}{2}}$ 
are unitary operators and the kernel of $e^{\frac{it\triangle}{2}}$ 
is given by  
$$
E(t,x,y)=\sum_{E_j}e^{\frac{-it{E_j}}{2}}\overline{u_j(x)}u_j(y)
$$
where $\{u_j(x)\}$ is eigenfunction expansion of $-\triangle$ 
and $E_j$ are eigenvalues.
The behavior of $E(t,x,y)$ is quite singular 
(See e.g. \cite{Ka}, \cite{Ni}, \cite{Ta 1}, \cite{Ta 2}, \cite{Ya 2}). 
Neverthless, when we sum a finite number of terms in $E$, $E_{finite}(t,x,y)$ are smooth 
and  we may intuitively choose classical shortest paths 
for low energy $E$. Accordingly we may define the heuristic approximation for 
Feynman path integration by $\{U_{\chi}(t/N)\}^N 
\rho({N^{1/3-\varepsilon}})$. 
Indeed, uniform convergences are proven in \S 3 :
\setcounter{section}{3}
\setcounter{theorem}{0}
\begin{proposition}
For small $\varepsilon>0$, 
$$
\lim_{N\rightarrow \infty} \Vert [\{U_{\chi}(t/N)\}^N -e^{\frac{it\triangle}{2}}] \rho(N^{1/3-\varepsilon}) \Vert_{L^2}=0.  
$$ 
\end{proposition} 
\par 
Another way to understand the low energy is WKB method  
in which well-known $h$-small semiclassical calculus gives  
the low energy good parametrices of $e^{it\triangle}$ 
(See for instance \cite[p.581]{BGT}, \cite{Ro}, \cite{Z}). 
So the low energy approximation is just a rewording WKB method 
in less $h$-small terminology. 
\par
For high energy functions, the strong convergence is assured by 
$L^2$ estimates in \S 4. Furthermore we shall see that this convergence 
is not uniform. 
\setcounter{section}{1}
\setcounter{theorem}{2}
\begin{remark}
In physics literature, Feynman is saying in his book \cite{FH} that each trajectory 
contributes to the total amplitude to go from $a$ to $b$ and 
that they contribute equal to the amplitude, but contribute 
at different phases. So the solution of 
Schr\"odinger equation should be given by ``sum over all possible trajectories of the system":  
$$
\phi (t, x)=\int_{\Gamma} e^{\frac{i}{\hbar}S_{t}(\gamma)} \phi (0, \gamma(0))\; {\mathcal D}\gamma
$$
In the formula above ${\mathcal D}\gamma$ denotes a Lebesgue-like measure on the path space $\Gamma$, 
$\hbar$ means the Planch constant and $S_{t}(\gamma)$ is the classical 
action functional of the system evaluated along the path $\gamma$. 
This is often superposition principle. Despite of the successfully predicting power of the concept 
of Feynman Path integral, the rigorous mathematical meaning of sum over all possible trajectries 
on the sphere is difficult to handle. 
For this reason, we suggest an alternative method of these caluculations on the sphere.
\end{remark}
%%%%%%%%%%%%%%%%%%%%%%%%%%%%%%%%%%%%%%%%%
\setcounter{section}{1}
\setcounter{theorem}{0}
\section{Notations and some prelimaries.}
\par
\hspace{5mm} In \S 2 , we shall give a short review of 
the notations and some prelimaries of Feynman path integral on $S^2$.  
%%%%%%%%%%%%%%%%%%%%%%%%%%%%%%%%%%%%%%%%%%%%%%%%%%%%%%%
\par 
We start out by rapidly going over a basic result on Van Vleck determinant. 
\par
Let $M$ be a Riemannian manifold and $x, y\in M$. 
If the geodesic $\stackrel{\textstyle\frown}{{xy}}$ is not null,  
the distance invariant, $V(t,x,y)$, is defined by 
$$
V(t,x,y)=g^{-1/2}(x)g^{-1/2}(y) \det _{ij}
\left(\frac{\partial^2 S}{{\partial {x_i}}{\partial {y_j}}}\right).
$$
\par 
$V(t, x, y)$ is called Van Vleck determinant 
(See e.g. \cite[p.220]{Ino-Ma}, \cite[p.17]{Wa}). 
On $S^2$, the following lemma is well-known (See e.g. \cite[p.3]{GSS}). 
We shall give an elementary proof in appendix. 
\begin{lemma}[Van Vleck determinant on $S^2$] Let $x, y\in S^2$ and $t\not=0$. 
We find 
$$
V(t, x, y)=\frac{d(x,y)}{t^2 \sin d(x,y)}\quad \text{for}\ 0\leqq d <\pi .
$$
\end{lemma}
%%%%%%%%%%%%%%%%%%%%%%%%%%%%%%%%%%%%%%%%%%%%%%%%%%%
We often abbreviate $d(x,y)$ to $d$ in the following sentences. 
\par
The smooth kernel of $U_{\chi}(t)$ can now be constructed as follows.  
%%%%%%%%%%%%%%%%%%%%%%%%%%%%
\begin{lemma}
Let $\widehat K(t, x, y)=\sqrt{V(t, x, y)} e^{iS(t,\ x,\ y)}
=\frac{1}{t}\sqrt{\frac{d(x, y)}{\sin d(x,y)}}
\exp\{\frac{i|d(x,y)|^2}{2t}\}$. Then we have 
$$
\left(i\frac{\partial}{\partial t}+\frac{1}{2}\triangle_x \right) 
\widehat K(t,x,y) 
=\left(i\frac{\partial}{\partial t}+\frac{1}{2}\triangle_y \right) 
\widehat K(t,x,y)=\frac{1}{t}\sqrt{\frac{d}{\sin d}}
\left(\frac{1}{8}+\frac{d^2-\sin^2 d}{8d^2 \sin^2 d}\right)
\exp\left(\frac{id^2}{2t}\right) 
\quad \text{for}\ 0\leqq d <\pi . 
$$
\end{lemma}
%%%%%%%%%%%%%%%%%%%%%%%%%%%%%%%%%%%%%%%%%
\begin{proof}
For fixed $y\in S^2$, we may assume 
$$
\begin{cases}
 &x=(\sin\theta \cos\phi,\ \sin\theta \sin\phi,\ \cos \theta)
   =(\sin d \cos\phi,\ \sin d \sin\phi,\ \cos d), \\ 
 &y=(0, 0, 1) 
\end{cases}\quad \cdots (y\mbox{-spherical coordinate}).
$$
In $y$-spherical coordinate on the sphere,  
$\triangle_x=\frac{1}{\sin d}\frac{\partial}{\partial d}\left(\sin d\frac{\partial}{\partial d}\right)
+\frac{1}{\sin^2 d}\frac{\partial^2}{\partial \phi^2}$. 
So
\begin{align*}
-\frac{1}{2}\triangle_x \widehat K(t,x,y) 
&= -\frac{1}{2} \left\{ \frac{1}{\sin d}\frac{\partial}{\partial d} \left(\sin d\frac{\partial}{\partial d}\right)
+\frac{1}{\sin^2 d}\frac{\partial^2}{\partial \phi^2}\right\} \left[\frac{1}{t}\sqrt{\frac{d}{\sin d}}
\exp\left(\frac{id^2}{2t}\right)\right] \\
&=-\frac{1}{2} \left\{ \frac{1}{\sin d}\frac{\partial}{\partial d} \left(\sin d\frac{\partial}{\partial d}\right)
\right\} \left[\frac{1}{t}\sqrt{\frac{d}{\sin d}}
\exp\left(\frac{id^2}{2t}\right)\right] \\
&=-\frac{1}{2\sin d}\frac{\partial}{\partial d}\left[ \left\{\frac{1}{2t}\sqrt{\frac{\sin d}{d}}
\left(\frac{\sin d-d\cos d}{\sin d}\right)+\frac{id}{t^2}\sqrt{d\sin d}\right\}\exp\left(\frac{id^2}{2t}\right)\right] \\
&=\left\{-\frac{i}{t^2}+\frac{d^2}{2t^3}-\frac{1}{8t}\left(1+\frac{d^2-\sin^2 d}{d^2 \sin^2 d}\right)\right\} 
\sqrt{\frac{d}{\sin d}} e^{\frac{id^2}{2t}}. 
\end{align*}
Summarizing the calculations, we have
$$
\begin{cases}
i\frac{\partial}{\partial t}\widehat K(t,x,y)=\left(-\frac{i}{t^2}
+\frac{d^2}{2t^3} \right) \sqrt{\frac{d}{\sin d}} e^{\frac{id^2}{2t}}, \\
-\frac{1}{2}\triangle \widehat K(t,x,y)
=\left\{-\frac{i}{t^2}+\frac{d^2}{2t^3}-\frac{1}{8t}\left(1+\frac{d^2-\sin^2 d}{d^2 \sin^2 d}\right)\right\} 
\sqrt{\frac{d}{\sin d}} e^{\frac{id^2}{2t}}. 
\end{cases}
$$
\end{proof}
%%%%%%%%%%%%%%%%%%%%%%%%%%%%%%%%%%%%%%%%%
We note that $\left(\frac{1}{8}+\frac{d^2-\sin^2 d}{8d^2 \sin^2 d}\right) 
\Big|_{d=0} = \frac{1}{6}$. 
Thus we need to modify $\widehat K(t, x, y)$ 
with Dewitt curvature $\frac{R}{12}=\frac{1}{6}$ : 
$$
K(t, x, y)\equiv \widehat K(t, x, y) \cdot e^{\frac{iRt}{12}}
=\frac{1}{t}\sqrt{\frac{d(x, y)}{\sin d(x,y)}}
\exp\left\{\frac{i|d(x,y)|^2}{2t}+\frac{iRt}{12} \right\}. 
$$ 
This modification is known in physics literature (See for instance \cite{De}, \cite[p.457]{Do}).
%%%%%%%%%%%%%%%%%%%%%%%%%%%%%%%%%%%%%%%%%%%%%%%%%%%%%%%%%%%%%%%%%
By lemma 2.2,  
\begin{align*}
\left(i\frac{\partial}{\partial t}+\frac{1}{2}\triangle_x \right)
K(t,x,y)&=\frac{1}{t}\sqrt{\frac{d}{\sin d}}
\left(\frac{1}{8}+\frac{d^2-\sin^2 d}{8d^2 \sin^2 d}-\frac{1}{6} \right)
\exp\left\{\frac{i|d(x,y)|^2}{2t}+\frac{iRt}{12} \right\} \\
&=\frac{1}{t}\sqrt{\frac{d}{\sin d}}
\left(\frac{d^2-\sin^2 d}{8d^2 \sin^2 d}-\frac{1}{24} \right)
\exp\left\{\frac{i|d(x,y)|^2}{2t}+\frac{iRt}{12} \right\}
\quad \text{for}\ 0\leqq d <\pi . 
\end{align*} 
For $\chi(d)K(t, x, y)$, we obtain  
\begin{align*}
&\left(i\frac{\partial}{\partial t}+\frac{1}{2}\triangle_x \right)
(\chi K(t,x,y)) \\
&=\chi \left(\frac{d^2-\sin^2 d}{8d^2 \sin^2 d}-\frac{1}{24} \right)K(t,x,y) 
+\frac{1}{2}(\triangle_x \chi)K(t,x,y)
+\frac{1}{2}\left(2\frac{\partial \chi}{\partial d} \frac{\partial}{\partial d} K(t,x,y)\right) \\
&=\Big[\chi \left(\frac{d^2-\sin^2 d}{8d^2 \sin^2 d}-\frac{1}{24} \right)+\frac{1}{2}(\triangle_x \chi) 
+\frac{\partial \chi}{\partial d} 
\left(\frac{\sin d-d\cos d}{2d \sin d}\right) \Big] K(t,x,y)  \\
&\quad + \frac{\partial \chi}{\partial d} 
\left(\frac{id}{t} \right)  K(t,x,y).
\end{align*}
Seeing this, we define the error integral $E_{\chi_1}(t)$ and $E_{\chi_2}(t)$ by 
$$
\begin{cases} 
&E_{\chi_1}(t) f(x) \equiv \frac{1}{2\pi i}\int_{S^2}
\big[\chi \left(\frac{d^2-\sin^2 d}{8d^2 \sin^2 d}-\frac{1}{24} \right)
+\frac{1}{2}(\triangle_x \chi) 
+(\frac{\partial \chi}{\partial d} ) 
\left(\frac{\sin d-d\cos d}{2d \sin d}\right) \big] K(t,x,y)f(y)\; dy, \\
&E_{\chi_2}(t) f(x) \equiv \frac{1}{2\pi i}\int_{S^2}(\frac{\partial \chi}{\partial d})
\left(\frac{id}{t} \right)  K(t,x,y)f(y)\; dy. \\
\end{cases}
$$
where $\big[\chi \left(\frac{d^2-\sin^2 d}{8d^2 \sin^2 d}-\frac{1}{24} \right)
+\frac{1}{2}(\triangle_x \chi) 
+(\frac{\partial \chi}{\partial d} ) 
\left(\frac{\sin d-d\cos d}{2d \sin d}\right) \big]|_{d=0}=0$. 
%%%%%%%%%  Stationary phase  %%%%%%%%%%%%%%%%%%%%%%%%%%%%%%%%%%%%
\par
In order to estimate $U_{\chi}(t)$, $E_{\chi_1}(t)$ and $E_{\chi_2}(t)$, 
%%%%%%%%%%%%%%%%%%%%%%%%%%%%%%%%%%%%%%%%%%%%%%%%%%%%%%%%%%%%%%%
we state the method of stationary phase where $S(t,x,y)$ is a quadratic form,  
which is convenient here (See \cite[Lemma 7.7.3]{Ho}). 
\begin{lemma}
Let $A$ be a symmetric non-degenerate matrix with ${\rm Im} A \geqq 0$. 
Then we have for every integer $k>0$ and integer $s>n/2$
\begin{align*}
\Big| \int_{\mathbf{R}^n} u(x) e^{\frac{i <Ax, x>}{2t}}\; dx
-(\det(& A/2\pi i t))^{-\frac{1}{2}}
\sum_{j=0}^{k-1}(-it/2)^j \langle A^{-1}D, D \rangle^j u(0)/j! \Big| 
\\
&\leqq C_k (\Vert A^{-1} \Vert t)^{n/2+k} 
\sum_{|\alpha|\leqq 2k+s} \Vert D^{\alpha}u \Vert_{L^2(\mathbf{R}^n)}, 
\quad \text{for}\ \ u(x)\in \mathcal{S}(\mathbf{R}^n) .
\end{align*}
\end{lemma}
The right hand side in the above lemma is just the Sobolev norm :   
$$\Vert \cdot \Vert_{H^{2k+s}(\Omega)}
=\sum_{|\alpha|\leqq 2k+s} \Vert D^{\alpha} \cdot \Vert_{L^2 (\Omega)}.$$   
Letting $A=I$ (unit matrix), we obtain stationary phase lemma 
in the polar coordinate system, 
\begin{corollary} Let $\chi(r)\in C_0^{\infty}(\mathbf{R}) $ be 
the bump function 
with compact support contained in $|r|<R$. 
\begin{align*}
\Big| \int_{0}^{2 \pi} \int_{\mathbf{R}} \chi(r) r\ u(r,\theta) 
e^{\frac{i r^2}{2t}}\; dr &d\theta
-( 2\pi i t)\sum_{j=0}^{k-1}(it\triangle_{flat}/2)^j  u(0)/j! \Big| 
\\
&\leqq \tilde C_k t^{k+1} \Vert \chi u \Vert_{H^{2k+2}(\Omega_R)}
\quad \text{for}\ \ u(x, y)\in {C}^{\infty}(\mathbf{R}^2, \mathbf{C}), 
\end{align*}
where $\Omega_R=\{(x, y)\in{\mathbf{R}^2}\ |\ r=\sqrt{x^2+y^2}<R \}$.  
\end{corollary}
%%%%%%%%%%%%%%%%%%%%%%%%%%%%%%%%%%%%%%%%%%%%%%%%%%%%%%%%%%%%%%
From corollary 2.4 we shall now estimate $U_{\chi}(t)$, $E_{\chi_1}(t)$ and $E_{\chi_2}(t)$. 
%%%%%%%%%%%%%%%%%%%%%%%%%%%%%%%%%%%%%%%%%%%%%%%%%%%%%%%%%%%%%%%%%
\begin{proposition} For $R<\pi$ and $x\in S^2$, we have
\begin{align*} 
&\Big| U_{\chi} (t) f(x)- f(x) \Big| \leqq C t \Vert (-\triangle+1)^2 f \Vert_{L^2(S^2)}, \\ 
&\Big| E_{\chi_1}(t) f(x) \Big|  \leqq C t \Vert (-\triangle+1)^2 f \Vert_{L^2(S^2)}  \\
& and \\
&\Big| E_{\chi_2}(t) f(x) \Big|  \leqq C t \Vert (-\triangle+1)^3 f \Vert_{L^2(S^2)}
\quad \text{for}\ \ f(x)\in C^{\infty}({S}^2). 
\end{align*}
%where $\Omega_{R_x}=\{ y\in S^2\ |\ d(x,y)<R \}$.
\end{proposition}
%%%%%%%%%%%%%%%%%%%%%%%%%%%%%%%%%%%%%%%%%%%%%%%%%%%%%%%%%%%%%%%
\begin{proof} Take $x$-centered coordinate and $\Omega_{R_x}=\{ y\in S^2\ |\ d(x,y)<R \}$. 
\begin{align*} 
\Big| U_{\chi} (t) f(x)- f(x) \Big| &= 
\Big|\frac{1}{2 \pi i}\int_{0}^{2 \pi} \int_{0}^R  \sin r \left\{ \chi(r) 
\frac{1}{t}\sqrt{\frac{r}{\sin r}} \exp\left( \frac{ir^2}{2t}+\frac{iRt}{12}\right) f(r,\theta) \right\}  dr d\theta -f(x) \Big|  \\
& = \Big|\frac{1}{2 \pi i t}\int_{0}^{2 \pi} \int_{0}^R \chi(r)  r \left\{ 
\sqrt{\frac{\sin r}{r}} \exp\left( \frac{ir^2}{2t}+\frac{iRt}{12}\right) f(r,\theta) \right\} dr d\theta -f(x) \Big|  \\
&\leqq 
\Big|\frac{1}{2 \pi i t}\int_{0}^{2 \pi} \int_{0}^R \chi(r) r \left\{ 
\sqrt{\frac{\sin r}{r}} \exp\left( \frac{ir^2}{2t}\right) f(r,\theta) \right\} dr d\theta -f(x) \Big| 
+|(1-e^{iRt/12}) f(x)|.  
\end{align*}
Note that $\frac{\sin r}{r}|_{r=0}=1$. By Corollary 2.4 and putting $k=1$, we have  
$$ \Big| U_{\chi} (t) f(x)- f(x) \Big| \leqq c_1 t \Vert \chi f \Vert_{H^{4}(\Omega_{R_x})}+|(1-e^{iRt/12}) f(x)|. $$
Similarly, $(\frac{r^2-\sin^2 r}{8r^2 \sin^2 r}-\frac{1}{24})|_{r=0}=0$, it follows that 
$$ \Big| E_{\chi_1} (t) f(x) \Big| \leqq c_2 t \Vert \chi f \Vert_{H^{4}(\Omega_{R_x})}+|(1-e^{iRt/12}) f(x)|. $$    
Remarking that $\frac{\partial{\chi}}{\partial r}=0$ on the neighborhood of $r=0$ and putting $k=2$
$$ \Big| E_{\chi_2} (t) f(x) \Big| \leqq c_3 t \Vert \chi f \Vert_{H^{6}(\Omega_{R_x})}+|(1-e^{iRt/12}) f(x)|. $$ 
%%%%%%%%%%%%%%%%%%%%%%%%%%%%%%%%%%%%%%%%%%%%%%%%%%%%%%%%%%%%%%%%
Thus we only prove   
\begin{align*}
&\Vert \chi f \Vert_{H^6_{flat}(\Omega_{R_x})} \leqq c_4 \Vert (-\triangle+1)^3 f \Vert_{L^2(S^2)}, \tag{1} \\ 
&|f(x)|\leqq c_5 \Vert (-\triangle+1)^3 f \Vert_{L^2(S^2)} \tag{2} \\ 
\end{align*}
on local charts. The second inequality (2) is Sobolev's imbedding theorem on compact Riemannian manifold 
(See e.g. \cite[Theorem 2.20]{Au}, \cite[Theorem 0.2]{Ke}).
So we shall explain the first inequality by using G\r{a}rding inequality of elliptic operators. 
%(See e.g. \cite[aaaaa]{Au}).  
\par 
Take one "atlas" $\mathcal{A}$. Making the change of variables $y=T(x)$ and 
using $dy= |\det T|' dx \leqq \epsilon dx$, 
$$
\Vert \chi f \Vert_{H^6_{flat}(\Omega_x)} 
$$
are equivalent under changing coordinates. 
Furthermore, comaring with flat and spherical metric and using $g_{flat} \sim g_{S^2}$ on small local charts,  
$$
\Vert \chi f \Vert_{H^6_{flat}(\Omega_x)} \leqq c_6 \Vert \chi f \Vert_{H^6(S^2)}.
$$
Let $\phi_i$ be a partition of unity associated to $\mathcal{A}$. 
Recall that $\chi$ is said to be $C^{\infty}$ if $\chi\circ x_i^{-1}\in C^{\infty}$, we find  
%%%%%%%%%%%%%
$$
\Vert (\phi_i \chi f)\circ x_i^{-1}\Vert_{H^6(S^2)}
=\Vert [ \chi \circ x_{i}^{-1}](\phi_i f )\circ x_{i}^{-1} \Vert_{{H^6}(S^2)}
\leqq C_i \Vert (\phi_{i}f) \circ x_{i}^{-1} \Vert_{H^6(S^2)} 
$$ 
and summing this equation on $i$ shows 
$\Vert \chi f \Vert_{H^6(S^2)} \leqq 
c_7 \Vert f \Vert_{H^6{(S^2)}}$ 
holds with $c_7 \equiv max_{i}C_{i}$. 
Summarizing the calculations, 
$$
\Vert \chi f \Vert_{H^6_{flat}(\Omega_{R_x})} \leqq c_8 \Vert f \Vert_{H^6(S^2)}. \eqno{(3)} 
$$
We apply G\r{a}rding inequality of elliptic operators to (3), 
$$
\Vert \chi f \Vert_{H^6_{flat}(\Omega_{R_x})} \leqq c_8 \Vert f \Vert_{H^6(S^2)}
\leqq C \Vert (-\triangle+1)^3 f \Vert_{L^2(S^2)}.  
$$
\end{proof}
%%%%%%%%%%%%%%%%%%%%%%%%%%%%%%%%%%%%%%%%%%%%%%%%%%%%%%%%%%%%%%%%%%
We used Sobolev spaces in the above theorem. 
One of the general technique is use the theory of eigenfunction expansion of $-\triangle$. 
On the unit sphere, any square-integrable function (resp. distribution) 
$f(r, \theta)$ can be expanded as a linear combination of these: 
$$f(r, \theta)=\sum\limits_{l=0}^{\infty}
\sum_{m=-l}^{l} a_{l,m} Y_{l,m}(r, \theta). $$
This expansion holds in the sense of convergence in $L^2$ 
(resp. in ${\mathcal D}'$). 
Specially if
$$
\Vert \{(-\triangle)^k+1\} f \Vert_{L^2(S^2)}^2=
|a_{0,0}|^2+\sum\limits_{l=0}^{\infty}\sum_{m=-l}^{l} 
\{l(l+1)\}^k |a_{l,m}|^2 <\infty
$$
then $f$ is in the Sobolev space $H^k(S^2)$. 
%%%%%%%%%%%%%%%%%%%%%%%%%%%%%%%%%%%%%%%%%%%%%%%%%%%%%%%%%%%%%%%%%%%
Moreover we can reformulate the spectral projectors by  
$$
\rho(E) f=\sum\limits_{l(l+1)<E}
\sum_{m=-l}^{l} a_{l,m} Y_{l,m}(r, \theta) 
\quad 
\mbox{for}\ f \in \mathcal{D}'(S^2).
$$
%%%%%%%%%%%%%%%%%%%%%%%%%%%%%%%%%%%%%%%%%%%%%%%%%%%%%%%%%%%%%%%%%%
\section{Feynman path integral on $S^2$ for low energy functions}
%%%%%%%%%%%%%%%%%%%%%%%%%%%%%%%%%%%%%%%%%%%%%%%%%%%%%%%%%%%%%%%%%
The purpose of this section is to show 
the products of $U_\chi$'s 
converge uniformly for low energy functions in $L^2$. 
We abbreviate $E_{\chi_1}(t)+E_{\chi_2}(t)$ to $E(t)$ in the following sentences. 
\setcounter{section}{3}
\setcounter{theorem}{0}
\begin{proposition}[Time slicing products and energy limits]
For small $\varepsilon>0$, we have 
$$
\lim_{N\rightarrow \infty} [\Vert \{U_\chi(t/N)\}^N 
-e^{\frac{it\triangle}{2}}]\rho({N^{1/3-\varepsilon}}) \Vert_{L^2}=0. 
$$
\end{proposition}
\begin{proof} For $f(x) \in C^{\infty}(S^2)$,  
$$
\left(i\frac{\partial}{\partial t}+\frac{1}{2}\triangle_x \right) 
U_{\chi}(t) f(x) =E(t) f(x).  
$$
This leads to 
$$
U_{\chi}(t) f(x) 
=e^{\frac{it\triangle}{2}}
\left(1+ \int_{0}^t e^{\frac{-is\triangle}{2}}  E(s) ds \right) f(x).   
$$
Here 
$
\Vert (\int_{0}^t e^{\frac{-is \triangle}{2}}  E(s) f(x) ds) \Vert_{L^2}
\leqq  \int_{0}^t \Vert e^{\frac{-is\triangle}{2}}  E(s) f(x) \Vert_{L^2} ds 
\leqq  \int_{0}^t \tilde C s \Vert (-\triangle+1)^3 f(x) \Vert_{L^2} ds 
\leqq  \frac{\tilde C t^2}{2} \Vert (-\triangle+1)^3 f(x) \Vert_{L^2}
$. 
Let $
\tilde E(t) f(x) \equiv \int_{0}^t e^{\frac{-is \triangle}{2}}  E(s) f(x) ds$.  $\tilde E(t) \triangle =\triangle \tilde E(t) $ and  
consequently we have 
\begin{align*}
\Vert \underbrace {e^{\frac{it \triangle}{2N}} e^{\frac{it \triangle}{2N}} \cdots e^{\frac{it \triangle}{2N}}}_{N-k\ \mbox{times}} 
\underbrace {\tilde E(t/N) \tilde E(t/N) \cdots \tilde E(t/N)}_{k\ \mbox{times}} f(x)\Vert_{L^2} 
\leqq \Big(\frac{\tilde{C}}{2}\Big)^k \Big(\frac{t}{N}\Big)^{2k} \Vert (-\triangle+1)^{3k} f(x) \Vert_{L^2}. 
\end{align*}
%%%%%%%%%%%%%%
The binomical coefficients bounds $\begin{pmatrix} N \\ k \end{pmatrix}\frac{1}{N^k}<\frac{1}{k!}$ 
yields the following estimates 
\begin{align*}
\Vert \{e^{it\triangle/2}-U_\chi(t/N)^n\}f(x) \Vert_{L^2} 
&=\Vert \left[e^{it\triangle/2}-\{e^{it\triangle/N}(1+\tilde E(t/N))\}^N \right] f(x) \Vert_{L^2} \\  
&\leqq \sum_{k=1}^N 
\begin{pmatrix} N \\ k \end{pmatrix} 
\Vert \{e^{i(N-k)t\triangle/N} \tilde E(t/N)^{k}\}f(x)  \Vert_{L^2} \\
&\leqq \sum_{k=1}^N 
\begin{pmatrix} N \\ k \end{pmatrix} 
\Big(\frac{\tilde{C}}{2}\Big)^k \Big(\frac{t}{N}\Big)^{2k} 
\Vert (-\triangle+1)^{3k} f(x) \Vert_{L^2} \\
&\leqq \sum_{k=1}^N 
\frac{1}{k!}
\Big(\frac{\tilde{C}t^2}{2N} \Big)^k 
\Vert (-\triangle+1)^{3k} f(x) \Vert_{L^2}. 
\end{align*}
By using $\Vert (-\triangle+1)^{3k} \rho(E) f(x) 
\Vert_{L^2}\leqq (E+1)^{3k} \Vert f(x) \Vert_{L^2}$,    
\begin{align*}
\Vert \{e^{it\triangle/2}-U_\chi(t/N)^N\} \rho(E) f(x) \Vert_{L^2} 
&\leqq \sum_{k=1}^N 
\frac{1}{k!}
\Big\{\frac{\tilde C(E+1)^3 t^2}{2N}\Big\}^k \Vert f(x) \Vert_{L^2} \\ 
&\leqq \Big[\exp \Big\{ \frac{\tilde C(E+1)^3 t^2}{2N} \Big\}-1 \Big] \Vert f(x) \Vert_{L^2} \\
&\leqq \frac{C_2(E+1)^3 t^2}{2N} \Vert f(x) \Vert_{L^2}. 
\end{align*}
Thus for small $\epsilon >0$, 
$$
\lim_{N\rightarrow \infty} \Vert [\{U_\chi(t/N)\}^N  
-e^{\frac{it\triangle}{2}}] 
\rho({N^{1/3-\varepsilon}}) \Vert_{L^2} \leqq 
\lim_{N\rightarrow \infty} \frac{C_2(N^{1/3-\varepsilon}+1)^3 t^2}{2N}=0. $$ 
\end{proof} 
\par 
\begin{remark}
We note that $s-\lim\limits_{E\rightarrow \infty} e^{\frac{it\triangle}{2}} \rho(E) f(x)
=e^{\frac{it\triangle}{2}} f(x)$, so 
$$
s-\lim_{N\rightarrow \infty} \{U_\chi(t/N)\}^N 
\rho({N^{1/3-\epsilon}}) f(x) 
=e^{\frac{it\triangle}{2}} f(x) \quad\mbox{for}\ \forall\; f(x)\in L^2(S^2).  
$$
In $\S 4$, we show the stronger result by substituing $\rho({N}) $ for $\rho({N^{1/3-\epsilon}}) $. 
\end{remark}
\par
\begin{remark}
Some Trotter-Kato formulas for Feynman's operational caluculus contain  
infinite many spectral projectors, however we used a spectral projector once only(See \cite{I-T}).  
\end{remark}
%%%%%%%%%%%%%%%%%%%%%%%%%%%%%%%%%%%%%%%%%%%%%%%%%%%%%%%%%%%%%%%%%%%%%%%%%%%%%%%%%%%%%%%
\begin{remark} Gevrey classes on real compact manifolds are studied by many mathematicians. 
Gevrey classes are Fr\'echet spaces, and by \cite[Corollary 8]{FM} we know  
$$
f(x)\ \in G_{\sigma}(S^2)\  \Longleftrightarrow \  
\forall t>0,\ \lim\limits_{k\rightarrow \infty}|{S_k}|_{S^2} \exp (tk^{1/\sigma})=0
$$
where $f(x)=\sum_{k=0}^{\infty} S_k(\omega)$ is the spherical expansion of $f$. 
For $f(x)\ \in G_{1/6}(S^2)$
\par
\begin{align*}
\Vert \{e^{it\triangle/2}-U_\chi(t/N)^N\}\sum_{l=0}^{\infty} S_l(\omega) \Vert_{L^2} 
&\leqq \sum_{l=0}^{\infty} \sum_{k=1}^N 
\frac{1}{k!}
\Big\{\frac{\tilde C(l(l+1)+1)^3 t^2}{2N}\Big\}^k \Vert S_l \Vert_{L^2} \\ 
&\leqq \sum_{l=0}^{\infty} \Big[\exp \Big\{ \frac{\tilde C(l(l+1)+1)^3 t^2}{2N} \Big\}-1 \Big] \Vert S_l \Vert_{L^2} \\
&\leqq \sum_{l=0}^{\infty} \Big[\exp \Big\{ \frac{\tilde C(l(l+1)+1)^3 t^2}{2N} \Big\}-1 \Big] C| S_l |_{S^2} \\ 
& \rightarrow \ 0 \quad \mbox{as}\ N\rightarrow \infty.  
\end{align*}
Here we used 
the eigenfunction growth 
$\Vert S_k \Vert_{L^2}\leqq C |S_k|_{S^2} $.  
\end{remark}
%%%%%%%%%%%%%%%%%%%%%%%%%%%%%%%%%%%%%%%%%%%%%%%%%%%%%%%%%%%%%%%%%%
%%%%%%%%%%%%%%%%%%%%%%%%%%%%%%%%%%%%%%%%%%%%%%%%%%%%%%%%%%%%%%%%%%
\setcounter{section}{3}
\setcounter{theorem}{1}
\section{Strong limits for high energy functions}
%%%%%%%%%%%%%%%%%%%%%%%%%%%%%%%%%%%%%%%%%%%%%%%%%%%%%%%%%%%%%%%%%
In this section, we have 
the strong but not uniform convergence of time slicing products. 
To do this, we introduce the $L^2$ estimates known as H\"ormander and Maslov's 
theorem (See e.g. \cite[Theorem 2.1.1]{S} for more details). 
\begin{lemma} Let $a\in C_0({\mathbf{R}^n})$ and assume that 
$\Phi \in C^\infty$ satisfies $|\nabla \Phi|\geqq c >0$ on supp\ $a$. 
Then for all $\lambda>1$, 
$$
\Big|\int_{\mathbf{R}^n} a(x)e^{i \lambda \Phi(x)}\; dx \Big| 
\leqq C_{N} \lambda^{-N}, \quad N=1, 2, \cdots
$$ 
where $C_N$ depends only on $c$ if $\Phi$ and $a$ belong to a bounded 
subset of $C^{\infty}$ and $a$ is supported in a fixed compact set. 
\end{lemma}
\begin{proof}
Given $x_0 \in \mbox{supp}\; a$ there is a direction $\nu \in S^{n-1}$ 
such that $|\nu \cdot \nabla \Phi|\geqq \frac{c}{2}$ on some 
ball centered at $x_0$. Thus, by compactness, we can choose a partition 
of unity $\alpha_j \in C_0^{\infty}$ consisting of a finite number 
of terms and corresponding unit vectors $\nu_{j}$ 
such that $\sum \alpha_j(x)=1$ on $\mbox{supp}\ a$ and 
$|\nu_j \cdot \nabla \Phi|\geqq \frac{c}{2}$ on $\mbox{supp}\ \alpha_j$. 
If we set $a_j(x)=\alpha_{j}(x)a(x)$, it suffices to prove that 
for each $j$ 
$$
\Big|\int_{\mathbf{R}^n} a_j(x)e^{i \lambda \Phi(x)}\; dx \Big| 
\leqq C_{N} \lambda^{-N}, \quad N=1, 2, \cdots
$$ 
After possibly changing coordinates we may assume that 
$\nu_j=(1, 0, \ldots, 0)$ which means that 
$|\partial \Phi/ \partial x_1| \geqq c/2$ on supp $a_j$. If we let 
$$
L(x, D)=\frac{1}{i\lambda \partial \Phi/ \partial x_1}
\frac{\partial}{\partial x_1}, 
$$
then 
$L(x, D)e^{i\lambda \Phi(x)}=e^{i\lambda \Phi(x)}$. 
Consequently, if 
$
L^* = L^*(x, D)=\frac{\partial}{\partial x_1}
\frac{1}{i\lambda \partial \Phi/ \partial x_1}
$
is the adjoint, then 
$$
\int_{\mathbf{R}^n} a_j(x)e^{i \lambda \Phi(x)}\; dx = 
\int_{\mathbf{R}^n} (L^*)^N a_j(x)e^{i \lambda \Phi(x)}\; dx.  
$$
Since our assumptions imply that $(L^*)^N a_j=O(\lambda^{-N})$, 
the results follows. 
\end{proof}
%%%%%%%%%%%%%%%%%%%%%%%%%%%%%%%%%%%%%%%%%%%%%%%%%%%%%%%%%%%%%%%%%%%
\begin{lemma}
Suppose that $\phi(x, y)$ is a real $C^{\infty}$ function satisfying 
the non-degeneracy condition 
$$
\det \left(\frac{\partial^2 \phi}{{\partial x_j}{\partial y_k}}\right)\not =0
$$
on the support $a(x, y)\in C_0^{\infty}(\mathbf{R}^2 \times \mathbf{R}^2)$. 
Then for $t>0$, 
$$
\Vert \int_{\mathbf{R}^2} e^{i\frac{\phi(x, y)}{2t}} a(x, y) f(y)\; dy \Vert_{L^2({\mathbf{R}}^2)}
\leqq C t \Vert f \Vert_{L^2({\mathbf{R}}^2)}. 
$$
where $C$ is indep. of $t$ and $f(x)$. 
\end{lemma}
\begin{proof}
We note that 
$$ 
\nabla_x[\phi(x, y)-\phi(x, z)]
=\left(\frac{\partial^2 \phi(x, y)}{{\partial x_j}{\partial y_k}}\right)(y-z)
+O(|y-z|^2). 
$$
By using a smooth partition of unity we can decompose 
$a(x, y)$ into a finite number of pieces each of which 
has the property that 
$$
|\nabla[\phi(x, y) - \phi(x, z)]|\geqq c|y-z| \quad \mbox{on supp}\ a,  \eqno(1)
$$
for some $c>0$. 
\par 
To use this we notice that 
$$
\Vert T_{t} f\Vert_2^2=\int\int K_t(y, z) f(y)\overline{f(z)}\ dy\; dz,   \eqno(2)
$$
where 
$$
K_{t}(y, z)=\int_{{\mathbf{R}}^2} e^{\frac{i}{t}[\phi(x, y) - \phi(x, z)]} a(x, y) \overline{a(x, z)} dx. 
$$
However, (1) and Lemma 4.1 imply that 
$$
|K_{t}(y, z)|\leqq C_N(1+\frac{1}{t}|y-z|)^{-N} \quad \mbox{for}\ \ \forall N. 
$$
Take $N=2$. By appling Shur test, 
the operator with kernel $K_t$ sends 
$L^2$ into itself with norm $O(t^2)$. This along with (2) yields 
$$
\Vert T_{t} f\Vert^2_{L^2({\mathbf{R}}^2)} 
\leqq C t^2 \Vert f \Vert^2_{L^2({\mathbf{R}}^2)}, 
$$
as desired. 
\end{proof}
\begin{lemma}
\begin{align*}
\Vert (\int_{0}^t e^{\frac{-is \triangle}{2}} E_{\chi}(s) f(x) ds) \Vert_{L^2}
\leqq 
C_1 t \Vert  f(x) \Vert_{L^2}  
+ {C_2 t^2} \Vert (-\triangle+1) f(x) \Vert_{L^2}
\end{align*} 
\end{lemma}
\begin{proof}
We shall use the partition of unity $\{\phi_i\}$ on $S^2$ 
with very small support $\mbox{diam}\ \phi_i <\epsilon$.
\par
If $d(\mbox{supp}(\phi_i) , \mbox{supp}(\phi_j))>R+2\epsilon$, 
$$
\phi_j(x) \{ E_{\chi}(t) (\phi_i(y) f(y))\}(x)=0. 
$$
So we may assume $\phi_i$ and $\phi_j$ are contained in one local chart.   
The same calculation for $E_\chi$ on local charts as lemma 4.2 implies 
$$
\Vert T_{t, i, j, k, l} f\Vert_2^2=\int\int K_{t, i, j, k, l}(y, z)
\{g^{1/2}(y)f(y)\} \overline{\{g^{1/2}(z)f(z)\}}\ dy\; dz,   
$$
where 
$$
K_{t, i, j, k, l}(y, z)=\int_{{\mathbf{R}}^2} e^{\frac{i}{2t}[d^2(x, y) - d^2(x, z)]} \phi_i(x) \overline{\phi_j(x)} 
\phi_{k}(y) \overline{\phi_{l}(z)}
a(x, y) \overline{a(x, z)} dx. 
$$
%%%%%%%%%%%%%%%%%%%%%%%%%%%%%%%%%%%%%%%%%%%%%%%%%%%%%%%%%%%%%%%%%%%%%%%%%%%%%
We give a simple explanation of the boundedness of $T_{t, i, j, k, l}$. 
Take one chart on Riemann sphere. Then the Fubini Study metric is 
$$
ds^2=\frac{4}{1+u^2+v^2}(du^2+dv^2). 
$$
This leads to $g(u, v)=\frac{16}{(1+u^2+v^2)^2}$ and $c_1<g(u,v)<c_2$. 
From lemma 2.1, 
$$
\det _{ij} \left(\frac{\partial^2 d^2}{{\partial {x_i}}{\partial {y_j}}}\right)
=g^{1/2}(x)g^{1/2}(y) \frac{d(x,y)}{\sin d(x,y)} 
\quad \text{for}\ 0\leqq d <\pi .
$$
We have 
$$
\det _{ij} \left(\frac{\partial^2 d^2}{{\partial {x_i}}{\partial {y_j}}}\right)>c \quad \text{for}\ 0\leqq d <R-2\epsilon.
$$
Applying lemma 4.2, we have $ \Vert T_{t, i, j, k, l} f\Vert_2<C $.  
$i, j$'s are finite and we conclude 
%%%%%%%%%%%%%%%%%%%%%%%%%%%%%
$$
\Vert E_{\chi_1}(t) f \Vert_{L^2} =\Vert \sum\limits_{i,j} \phi_i(x) 
\{E_{\chi_1}(t) \phi_j(y) f\}(x) \Vert_{L^2}<C_1.  \eqno{(3)}  
$$
%%%%%%%%%%%%%%%%%%%%%%%%%%%%%%
For $E_{\chi_2}$, we have 
\begin{align*}
E_{\chi_2}(t)f(x) 
&=\frac{1}{2\pi i}\int_{S^2}(\frac{\partial \chi}{\partial d} 
\cdot \frac{1}{d} ) 
\left(\frac{id^2}{t} \right) K(t,x,y)f(y)\; dy \\
&=\frac{\partial}{\partial t} 
\Big[\frac{1}{2\pi i}\int_{S^2}
(-\frac{\partial \chi}{\partial d}\cdot \frac{1}{d} ) 
\left\{t K(t,x,y)\right\} f(y)\; dy\Big] \\
\end{align*}
$(-\frac{\partial \chi}{\partial d}\cdot \frac{1}{d} )$ is bounded. So 
\begin{align*}
\Vert (\int_{0}^t e^{\frac{-is \triangle}{2}}  
E_{\chi_2}(s) f(x) ds) \Vert_{L^2}
&=\Vert \Big[e^{\frac{-is\triangle}{2}} 
\big\{
\frac{1}{2\pi i}\int_{S^2}
(-\frac{\partial \chi}{\partial d}\cdot \frac{1}{d} ) 
\left\{s K(s,x,y)\right\} f(y)\; dy \big\} \Big]_0^t\Vert_{L^2} \\
& \quad +\Vert \left(\int_{0}^t \frac{-is \triangle}{2} e^{\frac{-is \triangle}{2}}  
\big\{ \frac{1}{2\pi i}\int_{S^2}
(-\frac{\partial \chi}{\partial d}\cdot \frac{1}{d} ) 
\left\{s K(s,x,y)\right\} f(y)\; dy\big\} \; ds \right)\Vert_{L^2}
\\
& \leqq  \frac{C_2 t^2}{2} \Vert (-\triangle+1) f(x) \Vert_{L^2} \tag{4}. 
\end{align*}
Summarizing (3) and (4), we have 
\begin{align*}
\Vert (\int_{0}^t e^{\frac{-is \triangle}{2}}  
\{ E_{\chi_1}(s)+E_{\chi_2}(s) \} f(x) ds) \Vert_{L^2}
\leqq 
C_1 t \Vert  f(x) \Vert_{L^2}  
+ {C_2 t^2} \Vert (-\triangle+1) f(x) \Vert_{L^2}
\end{align*} 
%&=\Vert e^{\frac{it\triangle}{2}}
%\left(1+ \int_{0}^t e^{\frac{-is\triangle}{2}}  E(s) ds \right) 
%f(x) \Vert_{L^2} \\
%&\leqq \left(\Vert f(x)\Vert_{L^2}+ \int_{0}^t \Vert E(s) 
%f(x)\Vert_{L^2} ds \right) \\
%&\leqq (1+C|t|)\Vert f(x)\Vert_{L^2}. 
\end{proof}
%%%%%%%%%%%%%%%%%%%%%%%%%%%%%%%%%%%%%%%%%%%%%%%%%%%%%%%%%%%%%%%
It follows that $\{U_{\chi}(t/N)\}^N \rho(N)$ are uniformly bounded, 
so we have the strong limit. 
\begin{main}[Time slicing strong limits]
$$s-\lim_{N\rightarrow \infty} \{U_{\chi}(t/N)\}^N \rho(N) f(x)= e^{\frac{it\triangle}{2}} f(x)
\quad \mbox{for}\; \forall \; f(x) \in L^2(S^2). $$
\end{main}
\begin{proof}
By lemma 4.3, 
$\Vert U_{\chi}(t) \rho(E) f(x)\Vert \leqq \{1+C_1|t|+C_2 t^2 (E+1)\} \Vert f(x)\Vert_{L^2}$. 
Consequently 
$$
\Vert \{U_{\chi}(t/N) \rho(N) \}^N f(x)\Vert 
\leqq {(1+C_1|t|/N+C_2(N+1)t^2/N^2)}^{N}\Vert f(x)\Vert_{L^2}<e^{C|t|}\Vert f(x)\Vert_{L^2}. 
$$ 
The estimates of Proposition 3.1 yields 
\begin{align*}
\lim_{N\rightarrow \infty} \Vert (e^{\frac{it\triangle}{2}}-\{U_\chi(t/N)\}^N \rho(N) ) f(x) \Vert_{L^2}  
\leqq &
\lim_{N\rightarrow \infty} [\Vert e^{\frac{it\triangle}{2}} 
(1- \rho({N^{1/3-\varepsilon}}) ) f(x) \Vert_{L^2} \\
& + \Vert (e^{\frac{it\triangle}{2}} -\{U_\chi(t/N)\}^N)\ \rho({N^{1/3-\varepsilon}}) f(x) \Vert_{L^2} \\
&+\Vert \{U_\chi(t/N)\}^N\ (\rho(N) - \rho({N^{1/3-\varepsilon}})) f(x) \Vert_{L^2}] \\
=&0. 
\end{align*}
\end{proof}
\par
The counterexample for non-uniform convergence is 
derived from eigenvalue estimates for $U_{\chi}$. 
\par
For $x=(0,\ 0,\ 1) \in S^2$, we know 
$ Y_{l, 0}(0, 0)=\sqrt{\frac{2l+1}{2}} $ and 
$ Y_{l, m}(0,0)=0$ if $m\not=0$ . This leads to  
$$ \delta_{x}=\sum\limits_{l=0}^{\infty} \sqrt{\frac{2l+1}{2}} Y_{l,0} $$ 
in disribution sense. 
$U_{\chi}$ is well-defined on $\mathcal{D}'$. Remarking that 
$U_{\chi}$ and $-\triangle$ commute,  
$$ 
 U_{\chi}(\frac{t}{N}) \delta_{x}
=\frac{N}{2\pi i t}\chi(\theta) \sqrt{\frac{\theta}{\sin \theta}}e^{\frac{iN\theta^2}{2t}} 
=\sum\limits_{l=0}^{\infty} \alpha_l(t/N) \sqrt{\frac{2l+1}{2}} Y_{l,0}
$$
is square integrable with $L^2$ norm $CN/|t|$, where $\alpha_{l}(t/N)$ are eigenvalues of $U_{\chi}(t/N)$. 
$$
C^2(N/t)^2= \Vert \frac{N}{2\pi i t}\chi(\theta) \sqrt{\frac{\theta}{\sin \theta}}e^{\frac{iN\theta^2}{2t}} \Vert_{L^2}^2  
=\sum\limits_{l=0}^{\infty} {\frac{2l+1}{2}} |\alpha_l(t/N)|^2. 
$$
For each $N\in \mathbf{N}$, there exists $l_N$ satisfying $|\alpha_{l_N}(t/N)|<1/2$ and Riemann Lebesgue type estimate holds : 
$$
\Vert \{ U_{\chi}(\frac{t}{N}) \}^N Y_{l_N,0} \Vert_{L^2}=|\alpha_{l_N}(t/N)|^N <(1/2)^N. 
$$  
Immediately we have the following: 
\begin{theorem}(non-uniform convergence)
$$
\Vert e^{\frac{it\triangle}{2}}-\{U_\chi(t/N)\}^N \Vert > 1/2 
$$
\end{theorem} 
\par 
\begin{remark}
We can't apply the same method for $\mathbf{R}^n$ or non-trapping potentials.  
In this case, high energy eigenfunctions spread on whole spaces 
and Feynman path integral converges uniformly. 
The above theorem indicates that the 
weak concentration of high energy eigenfunctions  
gives Feynman path's strong convergence. 
\end{remark}
\begin{remark}
Since $S^2$ is compact, we need not to use 
Cotlar-Stein lemma (See e.g. \cite[p.238]{Fu 3}). 
\end{remark}
%%%%%%%%%%%%%%%%%%%%%%%%%%%%%%%%%%%%%%%%%%%%%%%%%%%%%%%%%%%%%%%%%%%%%%
%%%%%%%%%%%%%%%%%%%%%%%%%%%%%%%%%%%%%%%%%%%%%%%%%%%%%%%%%%%%%%%%%%
\setcounter{section}{4}
\setcounter{theorem}{1}
\section{Some remarks}
%%%%%%%%%%%%%%%%%%%%%%%%%%%%%%%%%%%%%%%%%%%%%%%%%%%%%%%%%%%%%%%%%
\begin{remark}
Our estimates hold in Sobolev spaces (See \S 2), that is  
$$
s-\lim_{N\rightarrow \infty} \{U_{\chi}(t/N)\}^N 
\rho({N^{1/3-\varepsilon}}) f(x) 
=e^{\frac{it\triangle}{2}} f(x) \quad \mbox{in}\ H^k(S^2).
$$ 
Remarking that $|\tilde E(t) f(x)| \leqq 
\int_{0}^t | E(s)\{e^{\frac{-is \triangle}{2}}f(x)\}| ds 
\leqq  \frac{\tilde C t^2}{2} \Vert (-\triangle+1)^3 f(x) \Vert_{L^2}
$, we also have uniformly convergence :
$$
\lim_{N\rightarrow \infty}\sup\limits_{x \in S^2}| [\{U_{\chi}(t/N)\}^N 
-e^{\frac{it\triangle}{2}}] \rho({N^{1/3-\varepsilon}}) f(x)|=0 
\quad \mbox{for}\ f(x)\in L^2(S^2).
$$ 
Carleson's theorem has been extended to zonal functions(See \cite[Theorem1.3]{Me}). 
If $f\in L^2(S^2)$ is zonal about a point $\tau \in S^2$, then its harmonic 
expansion is convergent everywhere : 
$$
\lim_{N\rightarrow \infty} \{U_{\chi}(t/N)\}^N 
\rho({N^{1/3-\varepsilon}}) f(x) 
=e^{\frac{it\triangle}{2}} f(x) 
\quad \mbox{a.e.}\ 
$$
\end{remark}
%%%%%%%%%%%%%%%%%%%%%%%%%%%%%%%%%%%%%%%%%%%%%%%%%%%%%%%%%%%%%%%%%%%%%%%
%%%%%%%%%%%%%%%%%%%%%%%%%%%%%%%%%%%%%%%%%%%%%%%%%%%%%%%
\begin{remark}
In physics literature(See e.g. \cite[p.2360]{GS}), Feynman path integral on $S^2$ is formally represented as  
\begin{align*}
\int_{\Omega} e^{\frac{i}{h} S(\gamma)} {\mathcal D}{\gamma}
&=K(t,x,y) \\
&=\sum_{l=0}^{\infty} \sum_{m=-l}^{l} e^{-ih^2 t l(l+1)/2} 
Y_{l,m}(\theta_1, \phi_1) Y_{l,m}(\theta_2, \phi_2) \\  
&=\sum_{l=0}^{\infty} \left( \frac{2l+1}{4\pi}\right) e^{-ih^2l(l+1)/2}  C_l^{1/2}(\cos d(x, y))                        
\end{align*}
where $\Omega$ is the path space on $S^2$, $C_l^{1/2}(x)$ are the Gegenbauer polynomials defined 
in terms of their generating function 
$$
\frac{1}{(1-2xt+t^2)^{1/2}}=\sum_{l=0}^{\infty} C_l^{1/2}(x)\; t^{l}.
$$
Last terms are interpreted in the sense of distributions.   
Take $h=1$ and $t=\frac{8\pi m}{k}\in {\mathbf Q}$ ($k$, $m$ are relatively prime), by M.Taylor's 
calculations (See \cite[p.147]{Ta 2}.) 
$$
e^{\frac{it \triangle}{2}}=e^{it/4 }\sum\limits_{j=0}^{2k-1} \Gamma (m, k, j)\cos \frac{2 \pi j}{k} A, 
$$
where $\Gamma (m, k, j)=\frac{1}{2 \pi} \sum\limits_{l=0}^{2k-1} e^{\pi i(l^2 m+lj)/k}$ are Gaussian sums and 
$\cos \frac{2 \pi j}{k} A$ are Fourier integral operators assocated to 
$A=\sqrt{-\triangle+\frac{1}{4}}$. We conclude 
\begin{align*}
s-\lim_{N\rightarrow \infty} \{U_{\chi}(8\pi m/kN)\}^N \rho(N) f(x)
&=\int_{S^2} \sum_{l=0}^{\infty} 
\left( \frac{2l+1}{4\pi}\right) e^{-i 8\pi m l(l+1)/2k}  C_l^{1/2}(\cos d(x, y))\ f(y) dy \\
&=\left\{e^{2  \pi i m/k} \sum\limits_{j=0}^{2k-1} \Gamma (m, k, j)\cos \frac{2 \pi j}{k} A\right\} f(x) 
\quad  \mbox{for}\ f(x) \in C^{\infty}(S^2). 
\end{align*}
The first term contains the imformations about geometry, the second term is a special function and 
the third term is denoted by Gaussian sums and F.I.O. How to establish the relationships explicit 
between these equations ?
\end{remark}
%%%%%%%%%%%%%%%%%%%%%%%%%%%%%%%%%%%%%%%%%%%%%%%%%%%%%%%%%%%%
\begin{remark}
We employed the shortest paths on $S^2$. 
$U_{\chi}(t)$ is defined by the action integrals, 
Van Vleck determinants and Dewitt curvature. 
Van Vleck determinants diverge at antipodal points, 
thus we ignore the long paths. 
\par
On $S^1$, however, we can take infinite many long paths 
for Fresnel integrable functions. 
Then all eigenfunctions are Fresnel integrable. 
On $S^2$, can one construct the analogy ?
\end{remark}
%%%%%%%%%%%%%%%%%%%%%%%%%%%%%%%%%%%%%%%%%%%%%%%%%%%%%%%%%%%%%%%%%%
\setcounter{section}{5}
\setcounter{theorem}{1}
\section{Conclusion}
%%%%%%%%%%%%%%%%%%%%%%%%%%%%%%%%%%%%%%%%%%%%%%%%%%%%%%%%%%%%%%%%%
\par
Simple WKB like formulas of Feynman integrations are discussed. 
Low energy approximations assure the unique classical paths. 
The quantum evolution is given by means of Van Vleck determinants 
and Dewitt curvature. That is 
$\{U_{\chi}(t/N)\}^N\ \mbox{proj}_N$ converges 
the Schr\"odinger operator in strong topology.  
Our strategy is available for general compact Riemannian manifolds,  
or super quadratic potentials on $\mathbf{R}$ {\cite{Mi}}.  
\vspace{3mm}\par 
\thanks{The author would like to thank Professor H.Nawa and 
Professor T.Suzuki for giving valuable comments. 
Professor A.Inoue encouraged me to try my studies again. 
I really appreciate his kindness. }
\vspace{5mm}\par 
%%%%%%%%%%%%%%%%%%%%%%%%%%%%%%%%%%%%%%%%%%%%%%%%%%%%%%%%%%%%%%%%%%
\setcounter{section}{5}
\setcounter{theorem}{1}
\appendix{\Large\bf Appendix}
%%%%%%%%%%%%%%%%%%%%%%%%%%%%%%%%%%%%%%%%%%%%%%%%%%%%%%%%%%%%%%%%%
\vspace{3mm}
\par In this appendix we present a direct calculation of Van Vleck 
determinant on $S^2$. 
\begin{proof}[Proof of Lemma 2.1]  
\par
Let $(\sin\theta\cos\phi,\ \sin\theta\sin\phi,\ \cos \theta)$ 
be a spherical coordinate system. 
\par 
For $
x=(\sin\theta_1\cos\phi_1,\ \sin\theta_1\sin\phi_1,\ \cos \theta_1) 
$ and 
$
y=(\sin\theta_2\cos\phi_2,\ \sin\theta_2\sin\phi_2,\ \cos \theta_2) 
$, 
\begin{align*}
           d&=\arccos (x \cdot y) \\
            &= \arccos(\sin\theta_1\cos\phi_1 \sin\theta_2\cos\phi_2 
               +\sin\theta_1\sin\phi_1 \sin\theta_2\sin\phi_2
               +\cos \theta_1 \cos \theta_2) \\
            &= \arccos\{\sin\theta_1 \sin\theta_2 \cos(\phi_1 -\phi_2)
                +\cos \theta_1 \cos \theta_2\}. \\
\end{align*}
\par
From definition 1.1 we have 
\begin{align*}          
 &S(t, x, y)=\frac{d^2}{2t}=\frac{1}{2t}\arccos^2 
                       \{\sin\theta_1 \sin\theta_2 \cos(\phi_1 -\phi_2)
                        +\cos \theta_1 \cos \theta_2\},  \\  
 &V(t, x, y)=g^{-1/2}(x)g^{-1/2}(y) 
\begin{vmatrix}
\frac{\partial^2 S}{\partial \theta_1 \partial \theta_2} & \frac{\partial^2 S}{\partial \theta_1 \partial \phi_2} \\
\frac{\partial^2 S}{\partial \phi_1 \partial \theta_2} & \frac{\partial^2 S}{\partial \phi_1 \partial \phi_2} 
\end{vmatrix} 
=\frac{1}{\sin \theta_1 \sin \theta_2} 
\begin{vmatrix}
\frac{\partial^2 S}{\partial \theta_1 \partial \theta_2} & \frac{\partial^2 S}{\partial \theta_1 \partial \phi_2} \\
\frac{\partial^2 S}{\partial \phi_1 \partial \theta_2} & \frac{\partial^2 S}{\partial \phi_1 \partial \phi_2} 
\end{vmatrix}. 
\end{align*}
%%%%%%%%%%%%%%%%%%%%%%%%%%%%%%%%%%%%%%%%%%%%%%%%%%%%%%
The caluculation of determinant is  
a little bit long. We give an intermediate expression 
checked by Mathematica 7.0. [Wolfram Research, Inc.] 
\begin{align*}
V(t, x, y)=\frac{1}{4t^2\sin \theta_1 \sin \theta_2}
&[16\arccos\{\cos\theta_1\cos\theta_2 
+ \cos(\phi_2 - \phi_1)\sin\theta_1\sin\theta_2\}\sin\theta_1\sin\theta_2]/ 
\\
&[10-2\cos 2(\phi_2 - \phi_1) + 
\cos 2(\phi_2 - \theta_1 - \phi_1) + \cos 2(\phi_2 + \theta_1 - \phi_1) 
\\ 
& -2\cos 2\theta_1\{1 + \{3 + \cos 2(\phi_2 - \phi_1)\} \cos 2\theta_2 \} 
\\
& - 4\cos 2\theta_2 \sin^2(\phi_2 - \phi_1) 
 -8\cos(\phi_2 - \phi_1)\sin 2\theta_1 \sin 2\theta_2]^{1/2} \\
%%%%%% 
=\frac{1}{4t^2}
&[16\arccos\{\cos\theta_1\cos\theta_2 
+ \sin\theta_1\sin\theta_2 \cos(\phi_2 - \phi_1)\}]/ 
\\
&[4\sin \arccos(\sin\theta_1\cos\phi_1 \sin\theta_2\cos\phi_2 
               +\sin\theta_1\sin\phi_1 \sin\theta_2\sin\phi_2
               +\cos \theta_1 \cos \theta_2) ] \\
=\frac{1}{t^2}
&[\arccos\{\cos\theta_1\cos\theta_2 
+ \sin\theta_1\sin\theta_2 \cos(\phi_2 - \phi_1)\}]/ 
\\
&[\sin \arccos\{\cos\theta_1\cos\theta_2 
+ \sin\theta_1\sin\theta_2 \cos(\phi_2 - \phi_1)\} ] \\
&\hspace{-8mm} =\frac{d(x, y)}{t^2 \sin d(x, y)}.
 \end{align*}
%%%%%%%%%%%%%%%%%%%%%%%%%%%%%%%%%%%%%%%%%%%%%%%%%%%%%%%%%%%%%%%%%%%%%%
\end{proof}
%%%%%%%%%%%%%%%

%%%%%%%%%%%%%%%%%%%%%%%%%%%%%%%%%%%%%%%%%%%%%%%%%%%%%%%%%%%%%%%%%%
%%%%%%% Affiliation.
\vspace{10mm}
\par 
Y.Miyanishi: Division of Mathematical Science, 
Department of System Innovation, 
\par\quad\quad\quad\quad\quad\quad
Graduate School 
of Engineering Science, Osaka University, 
\par\quad\quad\quad\quad\quad\quad
Machikaneyamacho 1-3, 
Toyonakashi 560-8531, Japan; 
\vspace{2mm}\par
e-mail: miyanishi@sigmath.es.osaka-u.ac.jp  
\end{document}